\documentclass[prd,aps,showpacs,nofootinbib,preprint]{revtex4-1}
\usepackage{amssymb,graphics,graphicx,color,epstopdf,subfigure,dcolumn,bm,slashed,tabularx}
\usepackage{amsfonts,amsmath,amsthm,amsxtra}
\usepackage{amsfonts,euscript,mathrsfs}

\newtheorem{theorem}{Theorem}

\makeatletter

\@ifundefined{textcolor}{}
{%
 \definecolor{BLACK}{gray}{0}
 \definecolor{WHITE}{gray}{1}
 \definecolor{RED}{rgb}{1,0,0}
 \definecolor{GREEN}{rgb}{0,1,0}
 \definecolor{BLUE}{rgb}{0,0,1}
 \definecolor{CYAN}{cmyk}{1,0,0,0}
 \definecolor{MAGENTA}{cmyk}{0,1,0,0}
 \definecolor{YELLOW}{cmyk}{0,0,1,0}
 }

\makeatother

\begin{document}

\title{Asymptotic Cosmological Behavior of Scalar-Torsion Mode in Poincar\'{e} Gauge Theory }

\author{Chao-Qiang Geng$^{1,2,3}$\footnote{geng@phys.nthu.edu.tw},
Chung-Chi Lee$^2$\footnote{g9522545@oz.nthu.edu.tw} and
Huan-Hsin Tseng$^2$\footnote{d943335@oz.nthu.edu.tw}
}

\affiliation{
$^1$College of Mathematics \& Physics, Chongqing University of Posts \& Telecommunications, Chongqing, 400065, China
\\
$^2$Department of Physics, National Tsing Hua University, Hsinchu, Taiwan 300
\\
$^3$Physics Division, National Center for Theoretical Sciences, Hsinchu, Taiwan 300
}

\date{\today}

\begin{abstract}
We study the cosmological effect of  the simple scalar-torsion ($0^+$) mode in Poincar\'{e} gauge theory
of gravity. We find that for the non-constant (affine) curvature case, the early evolution of the torsion density 
$\rho_T$ has a radiation-like asymptotic behavior of $a^{-4}$ with $a$ representing the scale factor, along 
with the stable point of the torsion pressure ($P_T$) and density ratio $P_T/\rho_T\rightarrow 1/3$ in the high redshift regime 
$(z \gg 0)$, which is different from the previous result in the literature. 
We use the Laurent expansion to resolve the solution. We also illustrate our result by the execution 
of numerical computations.
\end{abstract}

\maketitle

\section{Introduction}\label{sec:introduction}
The recent cosmological observations, such as those from type Ia
supernovae~\cite{obs1, obs11}, cosmic microwave background
radiation~\cite{obs12, arXiv:1001.4538}, large scale
structure~\cite{astro-ph/0501171, obs13} and weak
lensing~\cite{astro-ph/0306046}, reveal that our universe is subject
to a period of acceleration. In general,
there are two ways to explain the phenomenon of the late-time accelerating universe~\cite{DE}
either by modifying the left- and right-handed sides of Einstein equation,
called modified  gravity and modified matter theories, respectively.
 For modified gravity theories, the acceleration is accounted as
 a part of the gravitational effect, while modified matter theories
  are constructed by including some negative pressure matter that could result in
the expanding effect.
In this study, we adopt the view point of
an alternative gravity theory by selecting  so-called
Poincar\'{e} gauge theory (PGT)~\cite{Hehl:1976kj, Obukhov:1987tz,Hehl:1994ue},
which is also suitable to describe the late-time accelerating
behavior~\cite{Tseng:2012hz, Shie:2008ms}.

PGT is under the consideration of gauging
Poincar\'{e} group $P_4 = \mathbb{R}^{1,3} \rtimes O(1,3)$ for
gravity,
and it sets out from a Riemann-Cartan spacetime $(M,g,\nabla)$,
where $M$ is a differentiable manifold, $g$ is a metric on
$M$, and
$\nabla$ is a general affine metric-compatible  connection
with $\nabla g\equiv 0$
so that it has a canonical decomposition into
$\nabla = \overline{\nabla} + K$ with $\overline{\nabla}$  the
Riemannian part and $K$ the contortion tensor written in
terms of the torsion tensor $T$ of $\nabla$. As a result, PGT is in
general a theory of gravitation~\cite{Hehl:1976kj,Obukhov:1987tz} with torsion,
which  couples to the spin source. The theory comprises a degenerate case of
the Einstein's general relativity (GR) of the vanishing  torsion
$T$ and the Einstein-Cartan theory~\cite{Trautman:2006fp} with a
torsion field equation (TFE) algebraically coupled to the intrinsic
spin of the source, resulting in a non-dynamical torsion field.

In this work, we concentrate on a specific quadratic theory
with only the scalar-torsion mode
in PGT, which  possesses the dynamical torsion
field.
This particular mode is the \emph{simple $0^+$ mode}, which is
one of the six modes:
$0^{\pm},1^{\pm}$ and $2^{\pm}$  labeled by spin and parity, based on the linearized
theory~\cite{Hayashi:1981mm,Sezgin:1979zf}.
We remark that  the $0^-$ mode  
interacts with intrinsic spins of fermions~\cite{Shie:2008ms,Chen:2009at}. 
However, its contribution is considered to
have largely diminished from the early universe to the present time
so that its effect in the current stage must be slight. On the other
hand, since the $0^+$ mode has no interaction with any
fundamental source~\cite{Shie:2008ms,Kopczynski}, one could
imagine that it remains to have a considerable portion within the current
universe.
Consequently, this mode naturally becomes the subject to study~\cite{Hehl2012}.

In view of the FLRW  cosmology, the vanishing spin current 
of the scalar-torsion mode renders a
set of nonlinear equations which address the evolutions of the
metric and torsion field. Under the positivity energy argument,
Shie, Nester and Yo (SNY) in~\cite{Shie:2008ms} have observed two
separate cases: one has a constant affine curvature ($R \equiv$
const) which violates the positive energy condition; and the other
subject to the condition forms a system of nonlinear ordinary
differential equations (ODEs). In the former, it has a late time de
Sitter space asymptote~\cite{Ao:2010mg} yet the torsion energy
density could be negative, and the equation of state (EoS) of the
torsion field has an interesting behavior~\cite{Tseng:2012hz}. In
the latter, no obvious analytic solution is found so that numerical
methods are generally applied. In particular, one of interesting
features in our previous study~\cite{Tseng:2012hz}
for  the latter case  is that the torsion EoS  has an
asymptotic behavior in the high redshift regime, whereas some other
studies in the literature~\cite{Shie:2008ms,Ao:2010mg} point out that the affine
curvature $R$, torsion scalar $\Phi$ and  Hubble parameter $H=
\dot{a}/a$ are oscillatory during the cosmological evolution.
Since the oscillating behaviors do not  appear in our
work~\cite{Tseng:2012hz}, a thorough study is
 clearly needed. In
this paper, we find a proof from a semi-analytical solution in the
large curvature regime to support our non-oscillatory result.
In order to demonstrate the conformity with our semi-analytical
solution, we will also present the numerical analysis.


\section{Scalar-Torsion Cosmology}\label{TorsionCosmology}

\subsection{Formulation}

In this note, we explore a specific scalar-torsion mode in PGT called the \emph{simple
$0^+$ mode},
given by~\cite{Shie:2008ms}
\begin{equation}\label{E:SNY}
L_{\text{SNY}} = \frac{a_0}{2} R + \frac{b}{24} R^2 + \frac{a_1}{8}
\left( T_{ijk} \, T^{ijk} + 2 T_{ijk} \, T^{kji} - 4 T_k \, T^k
\right),
\end{equation}
with the positive coefficients of $a_{0,1}$ and $b$
are required by the positivity energy argument.
Under the FLRW metric, 
the field equations have been shown in Eqs.~(2.11)-(2.13) in Ref.~\cite{Tseng:2012hz}.
The Friedmann equations for the
scalar-torsion mode are given by
\begin{eqnarray}
H^2 &\equiv & \frac{\rho_c}{3a_0} \equiv{\rho_M + \rho_T\over 3a_0} ,\nonumber \\
 \dot{H} &\equiv & -\frac{\rho_c + p_{tot}}{2 a_0}\equiv -{\rho_c + p_M + p_T \over 2 a_0} ,
\label{E:H2}
\end{eqnarray}
where $a_0 = (8\pi G)^{-1}$ in GR, the subscript $M$ represents the
ordinary matter including both dust $(m)$ and radiation
$(r)$,
and $\rho_T$ and  $p_T$ correspond to the torsion density and pressure of the effective geometric
effect other than GR, defined as
\begin{eqnarray}
\label{E:rho,p_T}
\rho_T &=&  3\mu H^2 - \frac{b}{18} \left( R + \frac{6\mu}{b}
\right) (3H - \Phi )^2 + \frac{b}{24} R^2, \nonumber\\
p_T &=& \frac{\mu}{3}( R - \bar{R} ) + {\rho_T \over 3}\,,
\end{eqnarray}
respectively. From (\ref{E:rho,p_T}), we can discuss the torsion EoS, $w_T$, defined by~\cite{Tseng:2012hz}
\begin{equation}\label{E:w_T}
w_T = \frac{p_T}{\rho_T}\,.
\end{equation}


\subsection{Semi-analytical solution in high redshift}\label{sec:AnalyticalSolution}

In this context, we provide the semi-analytical solution of the
positive energy scalar-torsion mode in the large scalar affine
curvature limit $R \gg 6\mu/b$ which is commonly achieved in the
high redshift regime $(a \ll 1)$. In such situation, we  write
the energy density of ordinary matter and torsion in series
expansion of $a(t)$  as
\begin{eqnarray}
\label{E:rho_exp}
\rho_M &=& \frac{\rho_m^{(0)}}{a^3} +  \frac{\rho_r^{(0)}}{a^4},
\nonumber\\
\frac{\rho_T}{\rho_m^{(0)}} &= &\sum_{k=-c}^{\infty} A_{-k} \, a^k\,,
\end{eqnarray}
respectively.

Before the analysis, we first follow the rescaling of the parameters
in~\cite{Tseng:2012hz},
\begin{eqnarray}
\label{eq:norescaling} \tilde{a}_0 &=& a_0/m^2b, \quad \tilde{a}_1
=a_1/m^2b, \quad
\tilde{t} =t \cdot m, \quad  \tilde{\mu} =\tilde{a}_0 + \tilde{a}_1, \nonumber \\
\tilde{H}^2& =&H^2/m^2, \quad \tilde{\Phi} =\Phi/m, \quad
\tilde{R}=R/m^2,
\end{eqnarray}
where $m^2=\rho_m^{(0)}/3 a_0$. 
The equations of motion~\cite{Tseng:2012hz}
can be rewritten as dimensionless equations:
\begin{eqnarray}
\label{E:main eq1-n/dim}
&&\frac{d \tilde{H}}{d\tilde{t}}  = \frac{\tilde{\mu}}{6 \tilde{a}_1}  \tilde{R} - \frac{\tilde{a}_0}{2\tilde{a}_1 \, a^3} -2 \tilde{H}^2,\\
\label{E:main eq2-n/dim}
&& \frac{d\tilde{\Phi}}{d\tilde{t} } =
\frac{\tilde{a}_0}{2 \tilde{a}_1} \left( \tilde{R} - \frac{3}{a^3}
\right) - 3 \tilde{H} \tilde{\Phi} + \frac{1}{3} \tilde{\Phi}^2,\\
\label{E:main eq3-n/dim}
&& \frac{d \tilde{R}}{d\tilde{t}} \simeq
-\frac{2}{3} \tilde{R} \, \tilde{\Phi}, \\
\label{E:main eq4-n/dim}
&&\frac{\tilde{R}}{18} \left( 3\tilde{H} -
\tilde{\Phi} \right) - \frac{\tilde{R}^2}{24} - 3 \tilde{a}_1
\tilde{H}^2 = 3\tilde{a}_0 \left( \frac{1}{a^3} + \frac{\chi}{a^4}
\right),
\end{eqnarray}
respectively,
where $\chi = \rho^{(0)}_r /\rho^{(0)}_m $. In (\ref{E:main
eq3-n/dim}), we have taken the approximation of $R \gg 6\mu/b$ for the
high redshift regime. With the above rescaling, we shall argue that
the lowest order of $\rho_T$ does not exceed $a^{-4}$ in the
following discussion. We formulate the statement as a theorem.

\begin{theorem}
In the high redshift regime $(a \ll 1)$, $\rho_T = O(a^{-4})$.
\end{theorem}

\begin{proof}

First we expand
\begin{equation}\label{E:H^2}
\tilde{H}^2(t) = \sum^{\infty}_{k=-c} r_k \, a^{k-4},  \qquad ( r_k < \infty   )
\end{equation}
where $c$ is some integer, so that we have
\begin{equation}\label{E:dH}
\frac{d \tilde{H}}{d\tilde{t}} = \sum^{\infty}_{k=-c} \left(
\frac{k-4}{2} \right) r_k \, a^{k-4}.
\end{equation}
Using (\ref{E:main eq1-n/dim}), (\ref{E:main eq3-n/dim}), (\ref{E:H^2}) and (\ref{E:dH}),
we obtain
\begin{eqnarray}
\label{E:tilde_R} \tilde{R} &=& \frac{3 \tilde{a}_1 }{\tilde{\mu}}
\left( \sum^{\infty}_{k=-c} k \cdot r_k \,
a^{k-4} \right) + \frac{3 \tilde{a}_0 }{\tilde{\mu} \, a^3},\\
\label{E:tilde_Phi} \tilde{\Phi} &=& - \frac{3}{2} \tilde{H}  \cdot
\frac{\tilde{a}_1 \left( \sum^{\infty}_{k=-c} k(k-4) r_k \, a^{k-4}
\right) - \frac{3\tilde{a}_0}{a^3} }{\tilde{a}_1 \left(
\sum^{\infty}_{k=-c} k \, r_k \, a^{k-4} \right) +
\frac{\tilde{a}_0}{a^3}}\,.
\end{eqnarray}

Substituting (\ref{E:H^2}), (\ref{E:dH}), (\ref{E:tilde_R}) and
(\ref{E:tilde_Phi}) into (\ref{E:main eq2-n/dim}), and  comparing the
lowest power (requiring $c > -1 $, otherwise losing its leading
position) of $a$ in the high redshift, $a \ll 1$, we derive the
following relation
\begin{equation}\label{E:c}
\left( \frac{\tilde{a}_1}{\tilde{\mu}} \right)^2 c^2 \cdot r_{-c}^3
\left[ c^2 + (5-\frac{\tilde{a}_0}{\tilde{\mu}}) c + 4  \right] =0,
\end{equation}
which leads to $r_{-c}=0$ if  $c \geq 1$
as $0 <\tilde{a}_0/\tilde{\mu} < 1$ and $c^2 +(5-\frac{\tilde{a}_0}{\tilde{\mu}}) c + 4 \neq 0$.
This is equivalent to say that (\ref{E:H^2}) has
the form
\begin{equation}
\tilde{H}^2 = \frac{r_0}{a^4} + \frac{r_1}{a^3} + \frac{r_2}{a^2} +
\frac{r_3}{a} + r_4 + \cdots
\end{equation}

Finally, we achieve our claim from (\ref{E:H2}) that
\begin{eqnarray}
\label{E:rho_T}
\frac{\rho_T}{\rho_m^{(0)}} &= &- \left( \frac{\chi}{a^4} + \frac{1}{a^4} \right) + \tilde{H}^2
\nonumber\\
      & &= -\left( \frac{\chi}{a^4} + \frac{1}{a^3} \right) + \left( \frac{r_0}{a^4} + \frac{r_1}{a^3} + \frac{r_2}{a^2} + \frac{r_3}{a} + r_4 + \cdots
       \right) = O\left(\frac{1}{a^4}\right).
\end{eqnarray}
Note that the last equality follows since
$r_0\neq \chi$, which will be explained later. 
This is the end of the proof.
\end{proof}

We now write the expansion in (\ref{E:rho_T}), by
the theorem above, simply as
\begin{equation}\label{eq:rho_T}
\frac{\rho_T}{\rho_m^{(0)}} = \sum_{k=-4}^{\infty} A_{-k} \, a^k\,.
\end{equation}
We shall only take first few dominating terms for
a sufficient demonstration. By the procedure in the proof of
the theorem, we can as well compare terms of various orders to yield
the following relations,
\begin{eqnarray}
\label{eq:sol_pow3}
O(a^{-10})&:& 3 \left( A_4 + \chi \right) \left(1+ \frac{\tilde{a}_1}{\tilde{\mu}} \, A_3  \right)^2 =0, \\
\label{eq:sol_pow4}
O(a^{-9})&:& 2 \left( 1+ \frac{\tilde{a}_1}{\tilde{\mu}} A_3 \right) \left[ \frac{\tilde{a}_0}{\tilde{\mu}} \left( 1+ \frac{\tilde{a}_1}{\tilde{\mu}} A_{3} \right) A_{3} + \frac{4 \tilde{a}_1}{\tilde{\mu}} \left( A_4 + \chi \right) A_{2} \right] =0,\\
\label{eq:sol_pow5}
O(a^{-8})&:&  \left( 1+ \frac{\tilde{a}_1}{\tilde{\mu}} A_3 \right)  \left[ 4 \frac{\tilde{a}_0}{\tilde{\mu}} \left( 1+ 3
\frac{\tilde{a}_1}{\tilde{\mu}} A_3 \right) A_2 \right. \nonumber\\
&&\left.- \left( 3 A_2 +\frac{\tilde{a}_1}{\tilde{\mu}} \left( A_2 \left( 2 + 5 A_3 \right)
-18 A_1 \left( A_4+ \chi \right) \right) \right) \right] = 0.
\end{eqnarray}
 From  (\ref{eq:sol_pow3}), (\ref{eq:sol_pow4}) and
(\ref{eq:sol_pow5}), one concludes a relation,
\begin{equation}\label{eq:sol_a3}
A_3 = -\frac{\tilde{\mu}}{\tilde{a}_1} = - \frac{\left( \tilde{a}_0 + \tilde{a}_1 \right)}{ \tilde{a}_1} <-1,
\end{equation}
with $A_1,A_2$ and $A_4$ left as arbitrary constants to be
determined by initial conditions and (\ref{E:main eq4-n/dim}).
 Note that (\ref{eq:sol_a3}) implies $r_1= -\tilde{a}_0/ \tilde{a}_1 < 0$ 
 in (\ref{E:H^2}). However, due to the
observational data that  $a=1$ at the current stage, the radiation
density is much smaller than the dust density ($\chi \ll 1$), whereas the
torsion density is the same order as the dust density, as seen from
(\ref{E:rho_T}),
\begin{equation}
\frac{\rho^{(0)}_T}{\rho^{(0)}_m} = \left[ (r_0 - \chi) + r_2 +
\cdots \right] - ( 1+ |r_1|) \simeq O(1)\,.
\end{equation}
Subsequently, we have that $ [(r_0 - \chi) + r_2 + \cdots] \leq \max\{ O(1),
O(|r_1|) \}$, along with the assumption $r_k < \infty $ for each
$k$. As a result, we conclude that $r_k$, for all $k \neq 1$, should not be too large,
which forbids the possibility $r_0 = \chi$.
This argument shows the validity of the last equality in
(\ref{E:rho_T}) with the non-vanishing $O(1/a^4)$ coefficient.

From (\ref{E:w_T}), 
via the continuity equation~\cite{Tseng:2012hz}, we obtain
\begin{eqnarray}
\label{eq:sol_wt} w_T=-1-\frac{\rho^{\prime}_T}{3\rho_T} \simeq -1 +
\frac{1}{3}  \left( \frac{4 A_4 a^{-4} + 3 A_3 a^{-3}}{A_4 a^{-4} +
A_3 a^{-3}} \right) \simeq \frac{1}{3}\left( 1 - \frac{A_3}{A_4}a
\right),
\end{eqnarray}
where the prime $``\prime''$ stands for  $d/d\ln a$ and we have used
(\ref{eq:rho_T}) for $a \ll 1$.

\subsection{Numerical computations}\label{sec:Numerical}

In this subsection, we perform numerical computations to
support the analysis above. As an illustration, we take the parameters $\tilde{a}_0=2$
and $\tilde{a}_1=1$ and initial conditions
\[
\tilde{H}(z=0)=\tilde{H}_0=2, \qquad \tilde{R}(z=0) =
\tilde{R}_0=14,
\]
and show the evolutions of $w_T$, $\tilde{\Phi}$, and $\tilde{R}$ in
Figs.~\ref{fg:1}, \ref{fg:2}a and \ref{fg:2}b, respectively.
\begin{center}
\begin{figure}[tbp]
\begin{tabular}{ll}
\begin{minipage}{80mm}
\begin{center}
\unitlength=1mm \resizebox{!}{6.5cm}{\includegraphics{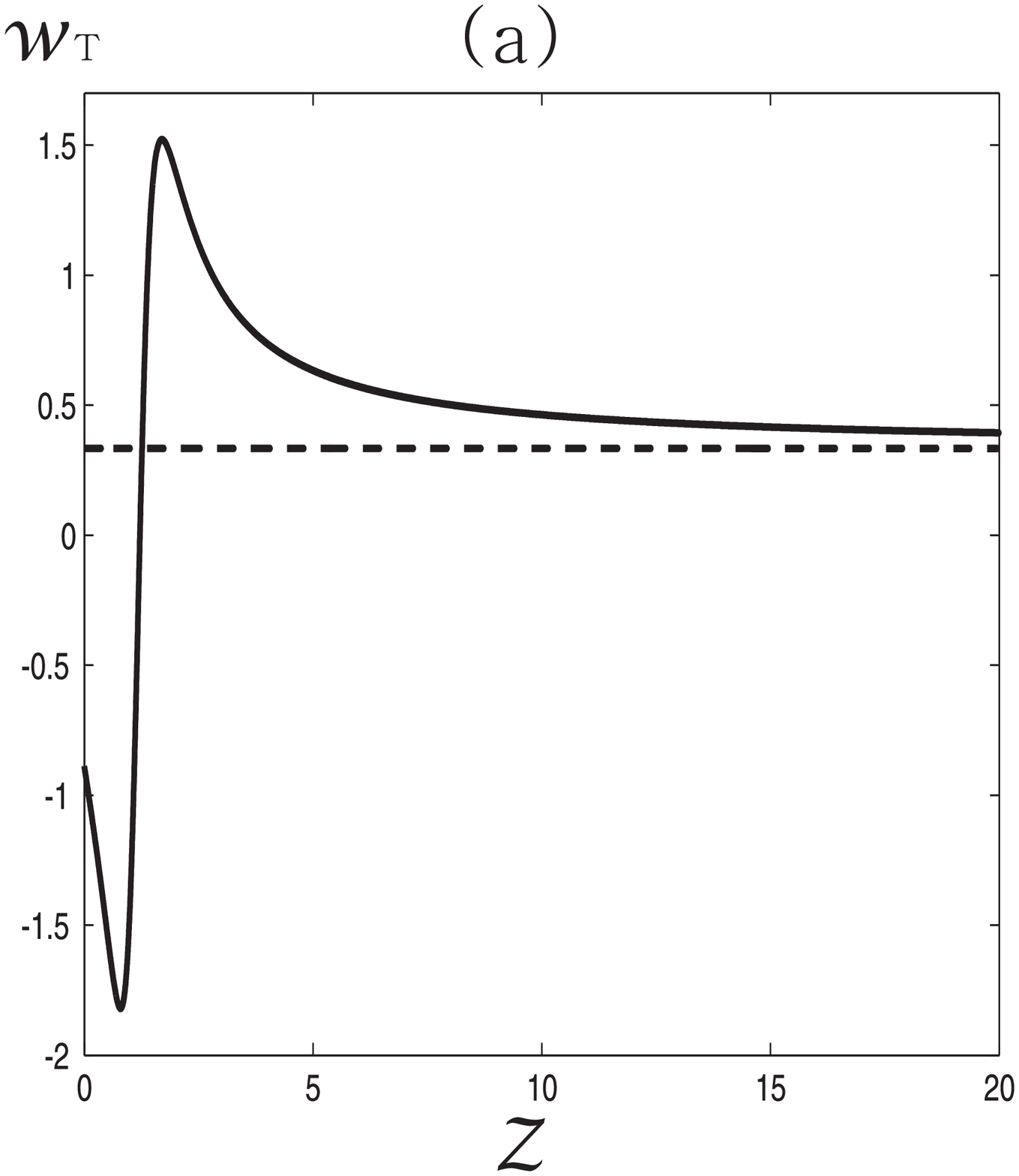}}
\end{center}
\end{minipage}
&
\begin{minipage}{80mm}
\begin{center}
\unitlength=1mm
\resizebox{!}{6.5cm}{\includegraphics{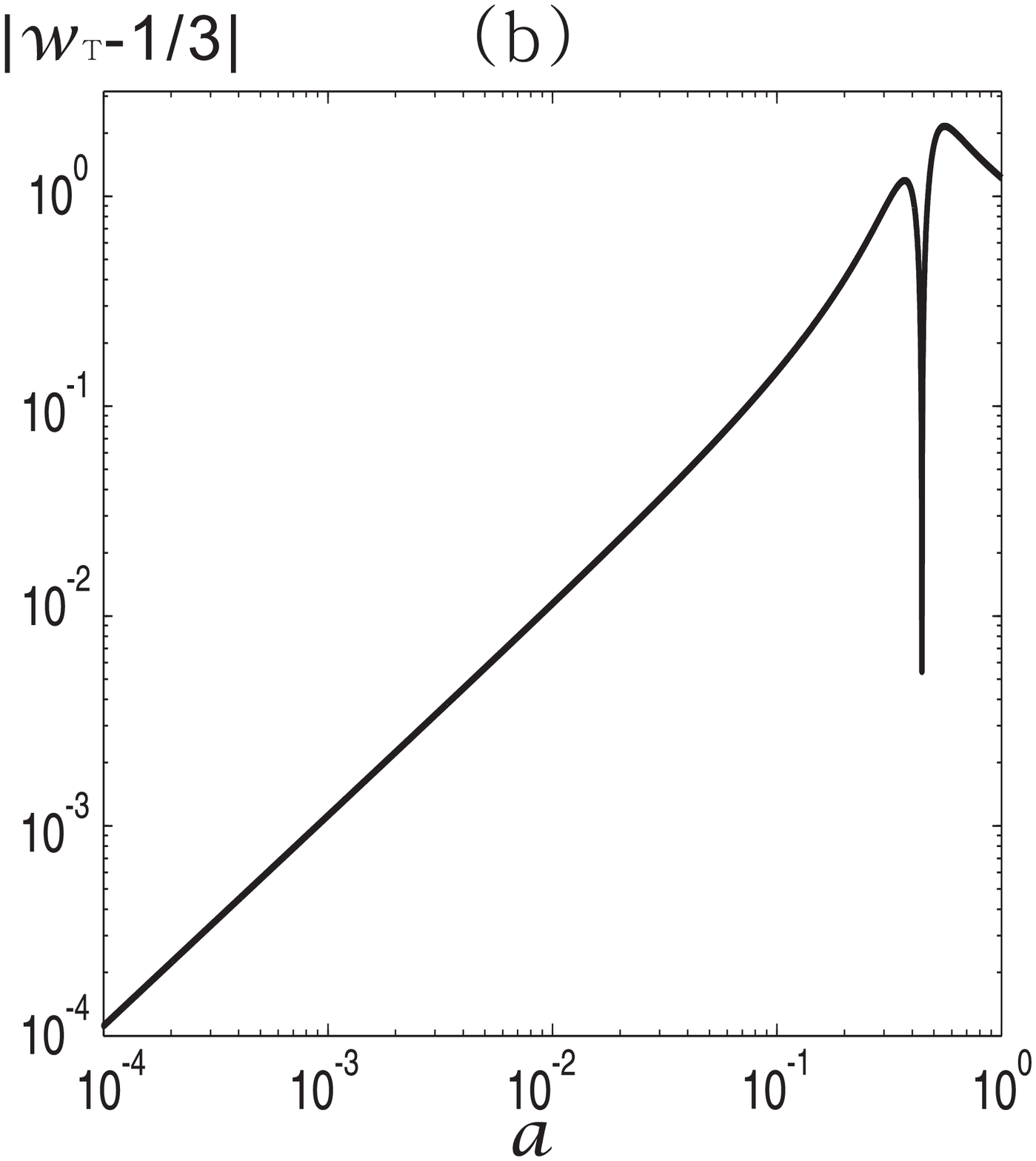}}
\end{center}
\end{minipage}\\[5mm]
\end{tabular}
\caption{Evolutions of (a)  $w_T$ and (b) $\lvert w_T-1/3 \rvert$ as
function of the redshift $z$  and the scale parameter $a$,
respectively, where the parameters and initial conditions are chosen
as $\tilde{a}_0=2$, $\tilde{a}_1=1$, $\tilde{H}_0=2$, $\tilde{R}_0=14$
and  $\chi = \rho^{(0)}_r /\rho^{(0)}_m = 3.1 \times 10^{-4}$.}
\label{fg:1}
\end{figure}
\end{center}

In Fig.~\ref{fg:1}a, we demonstrate the EoS of  torsion
 as a function of the redshift $z$. As seen from the figure,
in the high redshift regime $w_T$ approaches $1/3$, which indeed
shows an asymptotic behavior.
Fig.~\ref{fg:1}b indicates that $\lvert w_T-1/3 \rvert$
approximates a straight line in the scale factor $a$ in the
$\log$-scaled coordinate since the
slope in the $\log$-scaled coordinates is nearly $1$.
The singularity in the interval $[0.1,1]$
corresponds to the crossing $1/3$ of $w_T$.
Thus, the numerical results concur with our semi-analytical approximation
in (\ref{eq:sol_wt}).
In Fig.~\ref{fg:2}, we observe that the behaviors of $\tilde{R}$ and
$\tilde{\Phi} \propto 1/a^2$ in the high redshift regime are consistent
with the results in
(\ref{E:tilde_R}) and (\ref{E:tilde_Phi}), given by
\begin{eqnarray}
\label{eq:sol_R1}
&&\tilde{R} \simeq  \frac{2 \tilde{a}_1}{\tilde{\mu}} \, A_2 \, a^{-2}, \\
\label{eq:sol_phi1} &&\tilde{\Phi} \simeq 3\tilde{H} \propto a^{-2},
\end{eqnarray}
respectively, where $\tilde{H}^2 \simeq \left(\chi + A_4\right) a^{-4}$ from (\ref{E:H^2}).
Note that from (\ref{eq:sol_R1}), the behavior of the affine
curvature $\tilde{R}$ is highly different from that of the
 Riemannian scalar  curvature $\bar{R} = -
\mathcal{T}/a_0 = \rho_m/a_0$, which is proportional to $1/a^3$ in
both matter (dust) and radiation dominated eras.
\begin{center}
\begin{figure}[tbp]
\begin{tabular}{ll}
\begin{minipage}{80mm}
\begin{center}
\unitlength=1mm
\resizebox{!}{6.5cm}{\includegraphics{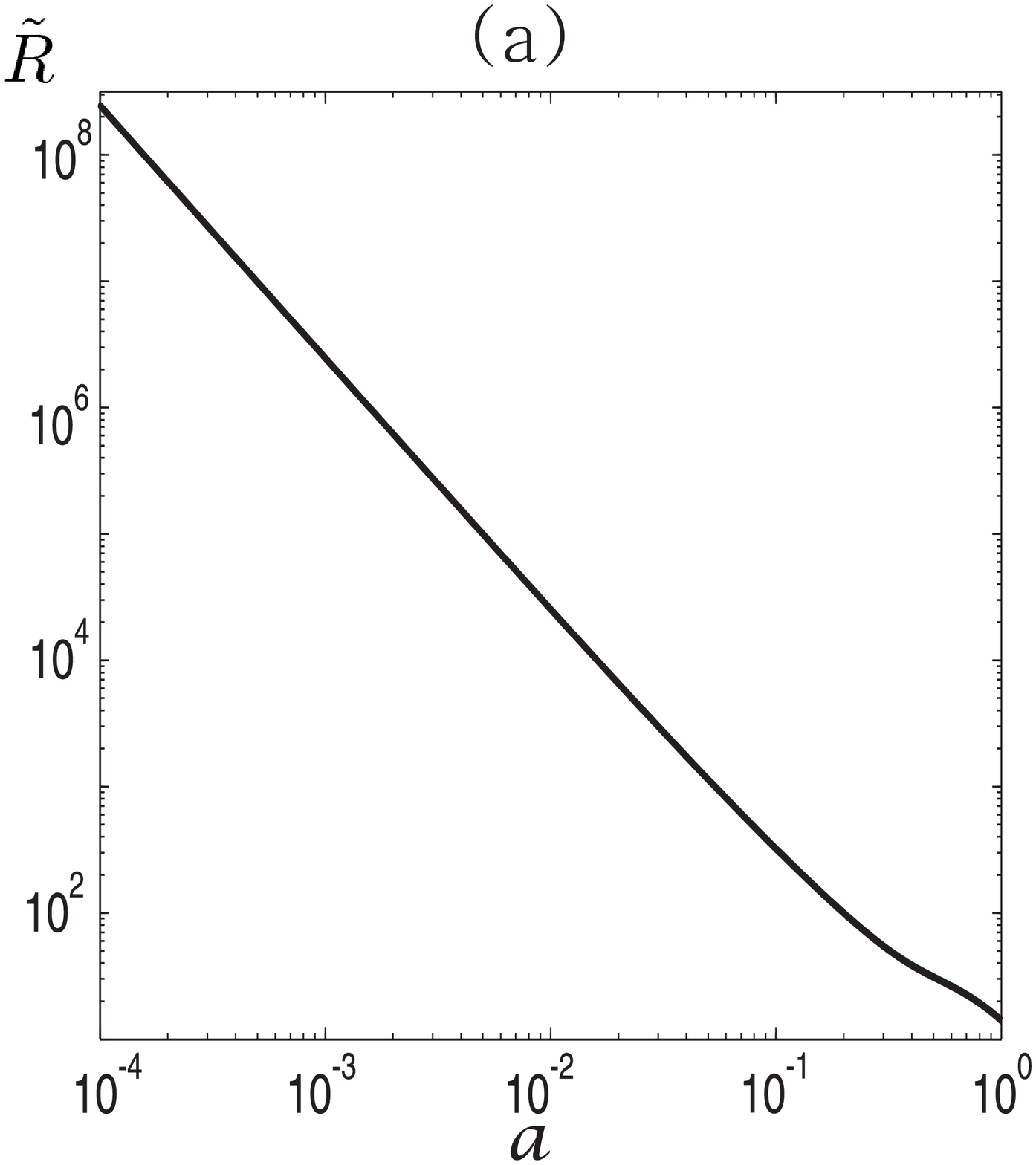}}
\end{center}
\end{minipage}
&
\begin{minipage}{80mm}
\begin{center}
\unitlength=1mm
\resizebox{!}{6.5cm}{\includegraphics{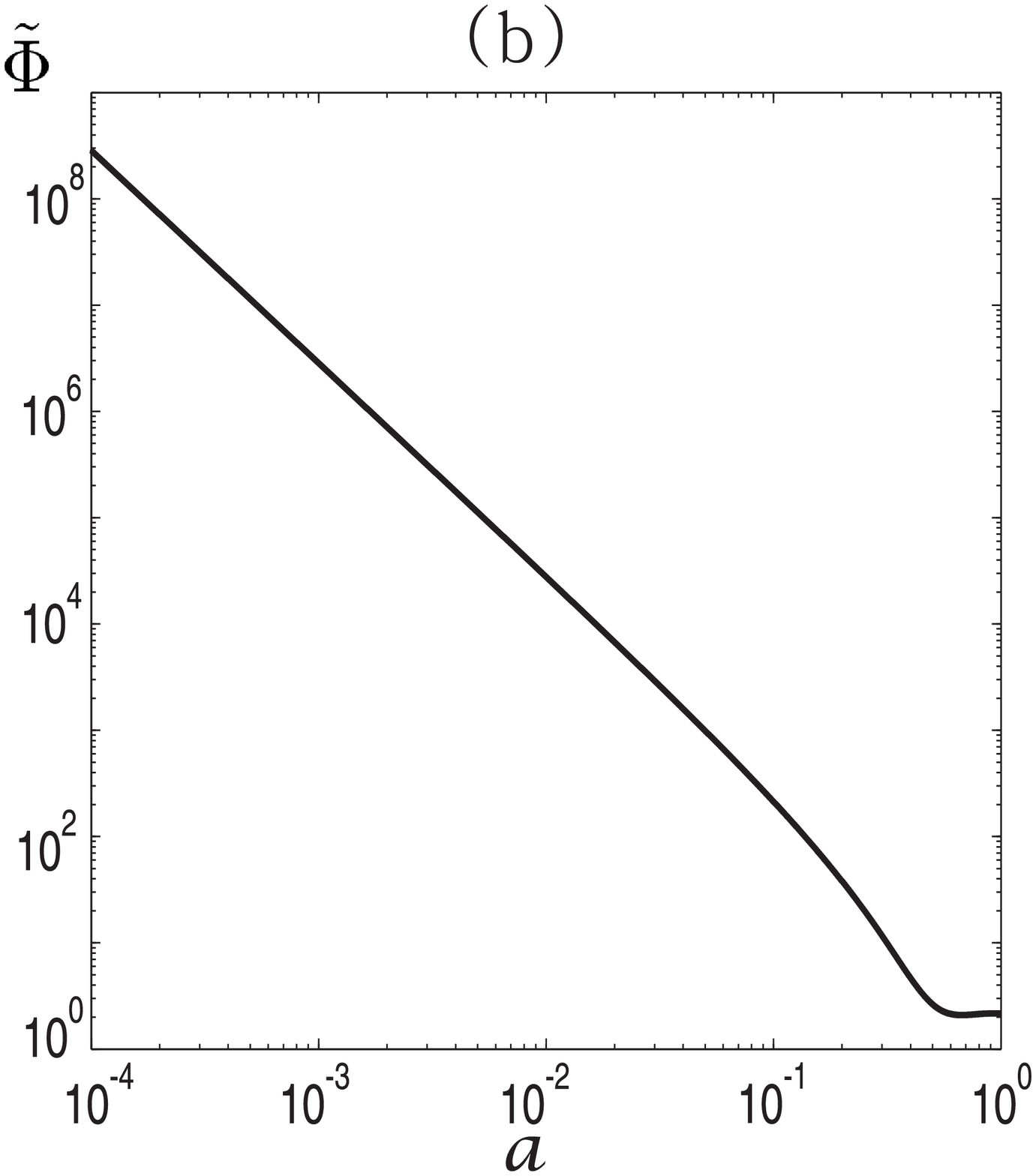}}
\end{center}
\end{minipage}\\[5mm]
\end{tabular}
\caption{ Evolutions of (a) the rescaled affine curvature
$\tilde{R}$ and (b) the torsion $\tilde{\Phi}$  as functions of the scale parameter
$a$ in the $\log$ scale with the parameters 
and initial conditions taken to be the same as Fig.~\ref{fg:1}.}
\label{fg:2}
\end{figure}
\end{center}

\section{Conclusions}\label{sec:conclusion}
We have investigated the asymptotic  evolution behaviors of the
scalar-torsion mode in PGT. The EoS of the torsion density has an
early time stable point $w_T(z \gg 0) \rightarrow 1/3$. This
behavior can be estimated through the semi-analytical solution via
the Laurent expansion in the scale factor $a(t)$ for the torsion density
$\rho_T$.
We have shown that there indeed exists the lowest
degree of $\rho_T$ in its expansion by $a^{-4}$, corresponding to
the radiation-like behavior in the high redshift regime. By the
comparison of the next leading-order term of $a^{-3}$ in the field
equations, we have extracted the coefficient $A_3 = -\mu / a_1$, which results
in the vanishing of the $a^{-3}$ term in the affine curvature $R$,
such that $R$ is only proportional to $a^{-2}$, consistent with the
numerical demonstration.\\

\noindent
{\bf Acknowledgments}

This work was partially supported by National Center for Theoretical
Science and  National Science Council (NSC-98-2112-M-007-008-MY3 and
NSC-101-2112-M-007-006-MY3) of R.O.C.

\end{document}